\documentclass{article}
\usepackage{graphicx,amssymb,amsmath}
\usepackage[utf8]{inputenc}
\usepackage{graphicx}
\graphicspath{./}

\newtheorem{theorem}{Theorem}
\newtheorem{cor}{Corollary}

\newtheorem{lemma}{Lemma}
\newenvironment{proof}{\paragraph{Proof:}}{\hfill$\square$}
\usepackage{xcolor}
\usepackage{todonotes} 
\usepackage[normalem]{ulem}
\usepackage[symbol]{footmisc}

\newtheorem{claim}{Claim}
\newtheorem{remark}{Remark}

\newcommand{\suggestion}[1]{{\color{blue} #1}}

\newcommand{\jyoti}[1]{{\color{black} #1}}
\newcommand{\attention}[1]{{\color{red} #1}}
\newcommand{\remove}[1]{{}}
\newcommand{\question}[1]{{\color{orange} #1}}

\newcommand{\Chek}{\it Check} 

\title{Compatible Paths on Labelled Point Sets\thanks{\jyoti{A preliminary version of the paper was presented at the 30th Canadian Conference on Computational Geometry (CCCG)~\cite{DBLP:conf/cccg/ArsenevaBBCCIJL18}}}}

\author{
Elena Arseneva\thanks{Universit\'e libre de Bruxelles (ULB), Belgium. {\tt ea.arseneva@gmail.com}}
\and 
Yeganeh Bahoo\thanks{Department of Computer Science,  University of Manitoba, Canada. {\tt bahoo@cs.umanitoba.ca}}
\and Ahmad Biniaz\thanks{School of Computer Science, University of Windsor, Canada. {\tt ahmad.biniaz@gmail.com }}
\and Pilar Cano\thanks{Department of Computer Science, Carleton University, Canada \& Universitat Polit\'ecnica de Catalunya, Spain {\tt m.pilar.cano@upc.edu}} 
\and  Farah Chanchary \thanks{Department of Computer Science, Carleton University, Canada {\tt  farah.chanchary@carleton.ca}} 
\and John Iacono\thanks{Universit\'e libre de Bruxelles, Belgium \& NYU, USA. {\tt jiacono@ac.ulb.be}}
\and Kshitij Jain\thanks{Cheriton School of Computer Science,  University of Waterloo, Canada. {\tt  \{k22jain,alubiw\}@uwaterloo.ca}}
\and Anna Lubiw\footnotemark[7] 
\and  Debajyoti Mondal\thanks{Department of Computer Science, University of Saskatchewan, Canada. {\tt dmondal@cs.usask.ca}}
\and  Khadijeh Sheikhan\thanks{NYU Tandon School of Engineering, Brooklyn, USA. {\tt khadijeh@nyu.edu}}
\and Csaba D. T\'{o}th\thanks{Department of Mathematics, California State University Northridge, Los Angeles, CA, USA. {\tt csaba.toth@csun.edu}} 
}

\begin{document}
\thispagestyle{empty}
\maketitle

\begin{abstract}
Let $P$ and $Q$ be finite point sets of the same cardinality in $\mathbb{R}^2$, each labelled from $1$ to $n$.
Two noncrossing geometric graphs $G_P$ and $G_Q$ spanning $P$ and $Q$, respectively, are called \emph{compatible} if for every face $f$ in $G_P$, there exists a corresponding face in $G_Q$ with the same clockwise ordering of the vertices on its boundary as in $f$.  
In particular, $G_P$ and $G_Q$ must be straight-line embeddings of the same connected $n$-vertex graph.
 
 \jyoti{Deciding whether two labelled point sets admit compatible geometric paths is known to be NP-complete.}
 We give polynomial-time algorithms to find compatible paths or report that none exist in three scenarios:
$O(n)$ time for points in convex position; 
$O(n^2)$ time for two simple polygons, where the paths are restricted to remain inside the closed polygons;
and $O(n^2 \log n)$ time for points in general position if the paths are restricted to be monotone.
\remove{
In this paper we examine some scenarios where compatible paths can be computed efficiently. If the input point sets are in convex position, we can compute compatible 
paths or certify that no such path exists in $O(n)$ time.  
For point sets in general position, we can 
compute compatible monotone paths or report that none exists 
in $O(n^2 \log n)$ time. 
}
\end{abstract}

\section{Introduction}
Computing noncrossing geometric graphs on finite point sets that are in some sense `compatible' is an active area of research in computational geometry. The study of compatible graphs is motivated by applications to shape animation and simultaneous graph drawing~\cite{baxter2009compatible,surazhsky2004high}. 

Let $P$ and $Q$ be finite point sets, each containing $n$ points in the plane labelled from 1 to $n$. Let $G_P$ and $G_Q$ be two 
noncrossing geometric graphs spanning $P$ and $Q$, respectively. $G_P$ and $G_Q$ are called \emph{compatible}, if for every face $f$ in $G_P$, there exists a corresponding face in $G_Q$ with the same clockwise ordering of the vertices on its boundary as in $f$. 
It is necessary, but not sufficient, that $G_P$ and $G_Q$ represent the same connected $n$-vertex graph $G$.
Given a pair of labelled point sets, it is natural to ask whether 
they have compatible graphs, and if so, to produce one such pair, $G_P, G_Q$. 
The question can also be restricted to specific graph classes such as paths, trees, triangulations, and so on; previous work (described below) has concentrated on compatible triangulations. 
Compatible triangulations of polygons are also of interest, which motivated us to examine compatible paths inside simple polygons.

In this paper we examine the problem of computing compatible paths on labelled point sets. 
Equivalently, we seek a permutation of the labels $1, 2, \ldots, n$ that corresponds to a noncrossing (plane) path in $P$ and in $Q$. 
Figures~\ref{fig:paths}(a)--(b) show a positive instance of this problem, and Figures~\ref{fig:paths}(c)--(d) depict an affirmative answer. 

\jyoti{Hui and Schaefer~\cite{DBLP:conf/isaac/HuiS04} proved the problem to be NP-complete. In this paper we develop some fast  polynomial-time algorithms to solve this problem under some restricted scenarios.}

\medskip
\noindent
{\bf Our results.}
We describe a quadratic-time dynamic programming algorithm that either finds compatible 
paths for two simple polygons, where the paths are restricted to remain inside the closed polygons, or reports that no such paths exist.
For the more limited case of two point sets in convex position, we give a linear-time algorithm to find compatible paths (if they exist). \jyoti{A linear-time algorithm for convex point sets  was also briefly outlined by Hui and Schaefer~\cite{DBLP:conf/isaac/HuiS04}. We  were not aware of this result during the preparation of the conference version~\cite{DBLP:conf/cccg/ArsenevaBBCCIJL18} of this paper.}
For two general point sets, we give an $O(n^2 \log n)$-time 
algorithm to find compatible monotone paths (if they exist).

\begin{figure*}[h]
\centering
\includegraphics[width=\textwidth]{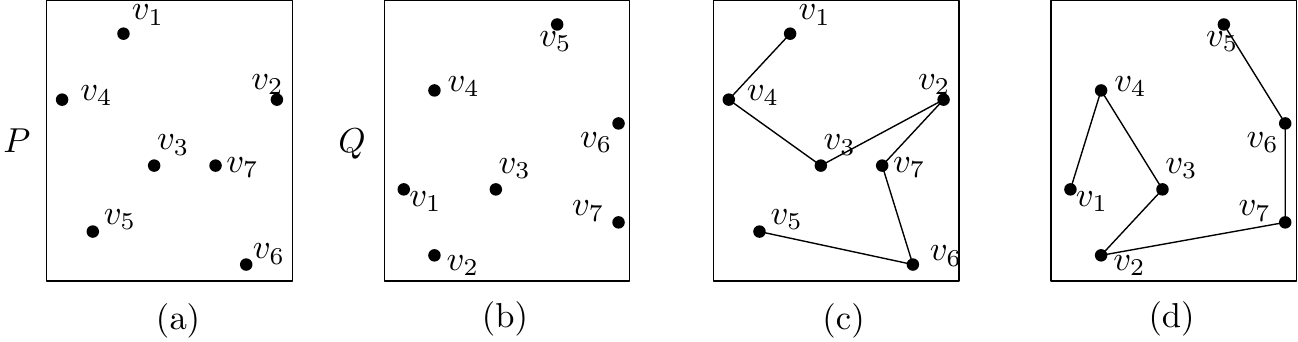}
\caption{(a)--(b)  A pair of labelled point sets $P$ and $Q$.  (c)--(d) A pair of compatible paths. 
}
\label{fig:paths}
\end{figure*}

\subsection{Background}
Saalfeld~\cite{DBLP:conf/compgeom/Saalfeld87} first introduced compatible triangulations of labelled point sets,  which he called ``joint'' triangulations. In  Saalfeld's problem, each point set is enclosed inside an axis-aligned rectangle, and the goal is to compute  compatible triangulations (possibly using Steiner points). Although not every pair of labelled point sets admit compatible triangulations, Saalfeld showed that one can always construct compatible triangulations using (possibly an exponential number of) Steiner points. 

Aronov et al.~\cite{DBLP:journals/comgeo/AronovSS93} proved that $O(n^2)$ Steiner points are always sufficient and sometimes necessary to compatibly triangulate two polygons when the vertices of the polygons are labelled $1, 2, \ldots, n$ in clockwise order. Babikov et al.~\cite{DBLP:conf/cccg/BabikovSW97} extended the $O(n^2)$ bound to \emph{polygonal regions} (i.e., polygons with holes), where the holes are also labelled `compatibly' (with the same clockwise ordering of labels). The holes may be single points, so this includes Saalfeld's ``joint triangulation'' problem. Pach et al.~\cite{DBLP:journals/algorithmica/PachSS96} gave an  $\Omega(n^2)$ lower bound on the number of Steiner   points in such scenarios. 

Lubiw and Mondal~\cite{Lubiw-Mondal-2017} proved that finding the minimum number of Steiner points is NP-hard for the case of polygonal regions.  The complexity status is open for the case of polygons, and also for 
 point sets. 
Testing for compatible triangulations without Steiner points may be an easier problem.
Aronov et al.~\cite{DBLP:journals/comgeo/AronovSS93} gave a polynomial-time dynamic programming algorithm 
to test whether two polygons admit compatible triangulations without Steiner points. But testing whether there are compatible triangulations without Steiner points is open for polygonal regions, as well as  
for point sets. 

The compatible triangulation problem seems challenging even for unlabelled point sets (i.e., when a bijection between $P$ and $Q$ can be chosen arbitrarily). Aichholzer et al.~\cite{aichholzer2003towards} conjectured that every pair of unlabelled point sets (with the same number of points on the convex hull) admit compatible triangulations without Steiner points. So far, the conjecture has been verified only for point sets with at most three interior points.


Let $G_S$ be a complete geometric graph on a point set  $S$. Let $H(S)$ be  the \emph{intersection graph} of the edges of $G_S$, i.e., each edge of $G_S$ corresponds to a vertex in $H(S)$, and two vertices are adjacent in $H(S)$ if and only if the corresponding edges in $G_S$ properly cross (i.e., the open line segments intersect). Every plane triangulation on $S$ has $3n{-}3{-}h$ edges, where $h$ is the number of points on the convex hull of $S$, and thus corresponds to a 
maximum 
independent set 
in $H(S)$. In fact, $H(S)$ belongs to the class of \emph{well-covered graphs}. (A graph is well covered if every maximal independent set of the graph has the same cardinality). A rich body of research attempts to characterize well-covered graphs~\cite{DBLP:journals/dam/FinbowHP17,DBLP:journals/jct/TankusT97}. Deciding whether two point sets, $P$ and $Q$, admit compatible triangulations is equivalent to testing whether $H(P)$ and $H(Q)$ have a common independent set of size $3n{-}3{-}h$. 


\section{Paths in Polygons and Convex Point Sets}

In this section we describe algorithms to find compatible paths on simple polygons and convex point sets. 
By compatible paths on polygons, we mean: given two polygons, find two compatible 
paths on the vertices of the polygons that are constrained to be non-exterior to the polygons.
(See Figures~\ref{fig:dp}(a)--(b).)
Note that convex point sets correspond to a special case, where the polygons are the convex hulls. 

\jyoti{\subsection{Negative Instances}}
Not every two convex point sets admit compatible paths, e.g., 5-point sets where the points are labelled (1,2,3,4,5) and (1,3,5,2,4), resp., in counterclockwise order. 

In this section we show that for every $n\ge 5$, there exist two convex labelled point sets, each containing $n$ points, that do not admit compatible trees. Note that this also rules out the existence of compatible paths.    

\begin{claim}  
\label{claim:tree-interval}
Let $P$ and $Q$ be point sets in convex position, each containing $n\geq 2$ points labelled by $\{1,2,\ldots , n\}$. If they admit a compatible tree that is not a star, then there exists a partition $\{1,2,\ldots ,n\}=A\cup B$ such that $2\leq |A|\leq |B|\leq n-2$ such that $A$ and $B$ are interval sets for both $P$ and $Q$.
\end{claim}
\begin{proof}
Suppose that $P$ and $Q$ admit a compatible tree $T$, which is not a star. Then $T$ has an edge $e$ between two vertices of degree two or higher. The deletion of $e$ decomposes $T$ into two subtrees, say $T_1$ and $T_2$, each with at least two vertices. The vertex sets of $T_1$ and $T_2$, resp., correspond to an interval set in $P$ and $Q$. 
\end{proof}

\begin{figure}[h]
\centering
\includegraphics[width=.6\textwidth]{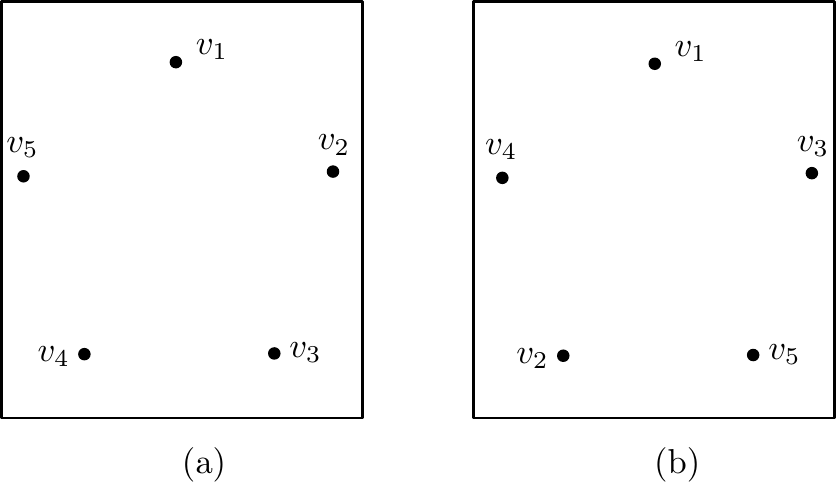}
\caption{Illustration for Lemma~\ref{lem:construction}. 
}
\label{fig:smallps}
\end{figure}

\begin{theorem}
\label{lem:construction}
For every integer $n\geq 5$, there exist two sets, $P_n$ and $Q_n$, 
each of $n$ labelled points in convex position, such that
    $P_n$ and $Q_n$ do not admit any compatible tree.
\end{theorem}
\begin{proof}
For $n=5$, let $P_5$ and $Q_5$ be point sets labelled $(1,2,3,4,5)$ and $(1,3,5,2,4)$, respectively, in counterclockwise order (Figure~\ref{fig:smallps}). 
%
%
If a compatible star exists, then the four leaves would appear in the same counterclockwise order in both $P_5$ and $Q_5$ (by the definition of compatibility). However, the two convex sets have distinct counterclockwise 4-tuples. If there is a compatible tree that is not a star, then by Claim~\ref{claim:tree-interval}, 
a 2-element set $A\subset \{1,2,3,4,5\}$ is an interval set for both $P_5$ and $Q_5$.
 However, all five consecutive pairs along the convex hull of $P_5$ are nonconsecutive in the convex hull of $Q_5$. Therefore, $P_5$ and $Q_5$ do not admit any compatible tree.  

For $n>5$, we can construct $P_n$ and $Q_n$ analogously. Let $P_n$ be labelled $(1,2\ldots ,n)$ in counterclockwise order. For $i=0,1,2,3,4$, let $N_i$ be the sequence of labels in $\{1,2,\ldots , n\}$ congruent to $i$ modulo 5 in increasing order. Now let $Q_n$ be labelled by the concatenation of the sequences $N_1,N_3,N_0,N_2,N_4$ in counterclockwise order.

If a compatible star exists, then the $n-1$ leaves would appear in the same counterclockwise order in both $P_n$ and $Q_n$ (by the definition of compatibility). However, the both neighbors of a vertex in $P_n$ are different from the two neighbors in $Q_n$, consequently $P_n$ and $Q_n$ do not share any counterclockwise $(n-1)$-tuple. If there is a compatible tree that is not a star, then by Claim~\ref{claim:tree-interval},  there is a partition $\{1,2,\ldots , n\}=A\cup B$ into interval sets, where $|A|,|B|\geq 2$. However, $A$ and $B$ cannot partition any subset of 5 consecutive elements in sequence $(1,2,\ldots , n)$, similarly to the case when $n=5$. Consequently, $P_n$ and $Q_n$ do not admit any compatible tree.
\end{proof}

\jyoti{\subsection{Algorithms}}
We first give a 
quadratic-time dynamic programming algorithm for simple polygons,
and then a linear time algorithm for convex point sets.


\begin{figure*}[h]
\centering
\includegraphics[width=.6\textwidth]{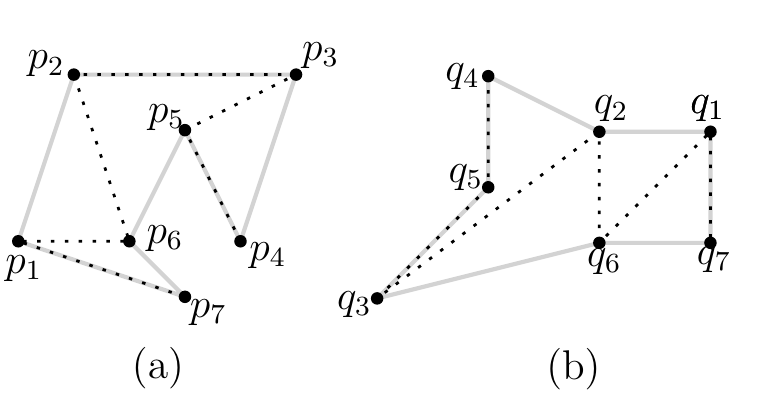}
\caption{Compatible paths on a pair of labelled polygons. The paths are drawn with dotted lines.}
\label{fig:dp}
\end{figure*}

We begin with two properties of any noncrossing path that visits all vertices of a simple polygon.
Let $P$ be a simple polygon with vertices $p_1, p_2, \ldots, p_n$ in some order (so the vertices have labels $1, 2, \ldots, n)$.
Let $\sigma$ be a label sequence corresponding to a noncrossing path that lies inside $P$ and visits all vertices of $P$. 
Define an \emph{interval} on $P$ to be a sequence of labels that appear consecutively around the boundary of $P$ (in clockwise or counterclockwise order).
For example, in Figure~\ref{fig:dp}(a), one interval is $(2, 1, 7, 6)$. 
Define an \emph{interval set} on $P$ to be the unordered set of elements of an interval.

\begin{claim}  
\label{claim:interval}
The set of labels of every prefix of $\sigma$ is an interval set on $P$.  
Furthermore, if the prefix does not contain all the labels, then the last label of the prefix corresponds to an endpoint of the interval.
\end{claim}
\begin{proof}
We proceed by induction on $t$, the length of the prefix, with the base case $t=1$ being obvious.  So assume the first $t-1$ 
labels form an interval set corresponding to 
interval $I$. 
Let $\ell$ be the $t$-th element of $\sigma$.
Suppose vertex $p_\ell$ is not contiguous with the interval $I$ on $P$.
Let $u$ and $v$ be the two neighbors of $p_\ell$ 
around the polygon $P$. 
Then $u$ and $v$ do not belong to $I$, and so the path must visit both of them after $p_\ell$.  
But then the subpath between $u$ and $v$ crosses the edge 
of the path that arrives at $p_\ell$,
contradicting the assumption that the path is noncrossing.
Thus vertex $p_\ell$ must appear just before or after $I$, 
forming a longer interval with $p_\ell$ as an endpoint of the interval.
\end{proof}

\begin{claim}
\label{claim:consecutive}
If $I$ is an interval on $P$ and $\sigma$ does not start or end in $I$, then the labels of $I$ appear in the same order in $\sigma$ and in $I$ (either clockwise or counterclockwise).  Note that the labels need not appear consecutively in $\sigma$.
\end{claim}

\begin{figure*}[h]
\centering
\includegraphics[width=.4\textwidth]{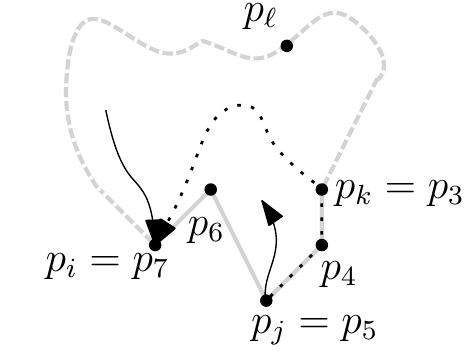}
\caption{Illustration for Claim~\ref{claim:consecutive}, where  $I=(p_7,p_6,p_5,p_4,p_3)$.  
}
\label{fig:dp-2}
\end{figure*}

\begin{proof}
Consider three labels $i,j,k$ that appear in this order in $I$.  Assume, for a contradiction, that these labels appear in a different order in $\sigma$ and suppose, without loss of generality, that they appear in the order $i,k,j$ in $\sigma$. Let $\ell$ be the last label of $\sigma$.  Because $\ell$ does not lie in $I$, the order of vertices around $P$ is $p_i,p_j,p_k,p_\ell$.  
See, e.g., 
Figure~\ref{fig:dp-2} where $i,j,k =7,5,3$.
Then the subpath of $\sigma$ from $p_i$ to $p_k$ crosses the subpath from $p_k$ to $p_\ell$, a contradiction.
%
\end{proof}

\subsubsection{An $O(n^2)$-time dynamic programming algorithm}

Let $P$, $Q$ be two $n$-vertex simple polygons with labelled vertices. 
Let $p_i$ (resp., $q_i$) be the vertex of $P$ (resp., $Q$) with the label $i$. 

Two vertices of a polygon are \emph{visible} if the straight line segment connecting the vertices lies entirely inside the polygon. We precompute the visibility graph of each polygon in $O(n^2)$ time~\cite{hershberger1989optimal}
such that later we can answer any visibility query in constant time.


\remove{  
\begin{claim}
\label{claim:interval}
For any compatible paths and for each polygon $P$, $Q$, the labels of any prefix correspond to a single interval $[i,j,cw]$ on the boundary of the polygon.  Furthermore, if the prefix does not contain all the labels, then the last label of the prefix corresponds to an endpoint of the interval.   
\end{claim}
\begin{proof}
We prove the claim for polygon $Q$. (The argument for $P$ is similar.)
If the prefix has length $n$, then it is straightforward to verify the claim. We thus assume that the prefix contains less than $n$ labels. 

The proof is by induction on $k$ with the base case $k=1$ being obvious.  So suppose the first $k-1$ labels correspond to 
an interval $I$
on the boundary of $Q$.  
Let $\ell$ be the last label in the prefix, so $q_\ell$ is the last vertex in the corresponding path in $Q$.
Suppose vertex $q_\ell$ is not contiguous with the interval $I$.
Let $u$ and $v$ be the predecessor and successor of $q_\ell$ in clockwise order around $Q$.  Then $u$ and $v$ do not belong to 
$I$ 
so the path must visit both of them later on.  But this is impossible for a crossing-free path, and compatible paths must be crossing-free. 

Thus vertex $q_\ell$ must appear just before or after 
$I$, 
forming a longer interval with $q_\ell$ as an endpoint of the interval.
%
\end{proof}
} 

Notation for our dynamic programming algorithm will be eased if we relabel so that polygon $P$ has labels $1, 2, \ldots, n$ in clockwise order.  For each label $i=1, \ldots, n$ and each length $t=1, \ldots, n$ let $I_Q(i,t,{\it cw})$ denote the interval on $Q$ of $t$ vertices that starts at $q_i$ and proceeds clockwise.  Define $I_Q(i,t,{\it ccw})$ similarly, but proceed counterclockwise from $q_i$.  Define $I_P(i,t,{\it cw})$ and $I_P(i,t,{\it ccw})$ similarly.  Note that  $I_P(i,t,{\it cw})$ goes from $p_i$ to $p_{i+t-1}$ (index addition modulo $n$).

We say that a path \emph{traverses} interval $I_Q(i,t,d)$ (where $d = {\it cw}$ or $\it ccw$), if the path is noncrossing, lies inside $Q$, visits exactly the vertices of $I_Q(i,t,d)$ and ends at $q_i$. 
We make a similar definition for a path to traverse an interval $I_P(i,t,d)$.

Our algorithm will solve subproblems $A(i,t,d_P, d_Q)$ where $i$ is a label from $1$ to $n$, $t$ is a length from $1$ to $n$, and $d_P$ and $d_Q$ take on the values ${\it cw}$ or $\it ccw$.  This subproblem records whether there is a path that traverses $I_Q(i,t,d_Q)$ and a path with the same sequence of labels that traverses $I_P(i,t,d_P)$.
If this is the case, we say that the two intervals are \emph{compatible}.
Observe that $P$ and $Q$ have compatible paths if and only if $A(i,n,d_P,d_Q)$ is true for some $i, d_P, d_Q$.

We initialize by setting $A(i,1,d_P,d_Q)$ to TRUE for all $i, d_P, d_Q$, and then solve subproblems in order of increasing $t$.
In order for intervals $I_Q(i,t+1,d_Q)$ and $I_P(i,t+1,d_P)$ to be compatible, the intervals of length $t$ formed by deleting the last label, $i$, must also be compatible, with an appropriate choice of direction ({\it cw} or {\it ccw}) on those intervals.  There are two choices in $P$ and two in $Q$.  We try all four combinations.  For a particular combination to `work' (i.e., yield compatible paths for the original length $t+1$ intervals), we need the last labels of the length $t$ intervals to match, and we need appropriate visibility edges in the polygons for the last edge of the paths.  

\begin{figure*}[h]
\centering
\includegraphics[width=.5\textwidth]{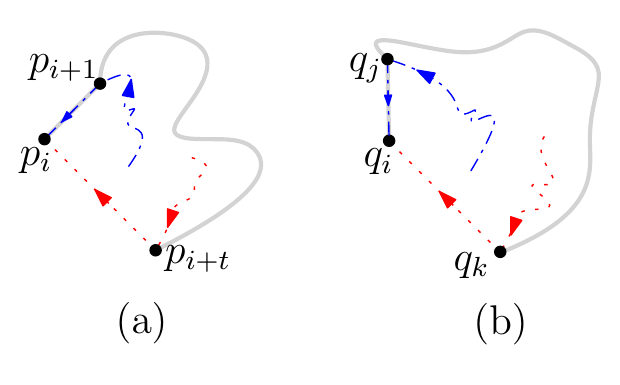}
\caption{Illustration for the dynamic programming algorithm.}
\label{fig:dp-3}
\end{figure*}

We give complete details for $A(i,t+1,{\it cw}, {\it cw})$. See Figure~\ref{fig:dp-3}.
(The other four possibilities are similar.)  
Deleting label $i$ from $I_P(i,t+1,{\it cw})$ gives $I_P(i+1,t, {\it cw})$ and $I_P(i+t, t, {\it ccw})$. 
Let $q_j$ be the vertex following $q_i$ in clockwise order around $Q$ and let $q_k$ be the other endpoint of $I_Q(i,t+1,{\it cw})$ (in practice, for efficiency, we would store $k$ with the subproblem).
Deleting label $i$ from $I_Q(i,t+1,{\it cw})$ gives $I_Q(j, t, {\it cw})$ and $I_Q(k, t, {\it ccw})$. The two possibilities for $P$ and $Q$ are shown by blue dash-dotted and red dotted lines in Figures~\ref{fig:dp-3}(a) and~(b), respectively. 
We set $A(i,t+1,{\it cw}, {\it cw})$ TRUE if any of the following four sets of conditions hold:
\begin{enumerate}
    \item Conditions for $I_P(i+1,t, {\it cw})$ and $I_Q(j, t, {\it cw})$:  $i+1=j$ and $A(i+1, t, {\it cw}, {\it cw})$.
    \item Conditions for $I_P(i+1,t, {\it cw})$ and $I_Q(k, t, {\it ccw})$: $i+1=k$ and $q_k$ sees $q_i$ in $Q$ and $A(i+1, t, {\it cw}, {\it ccw})$.  Note that the last edge of the path in $Q$ must be $(q_k, q_i)$ which is why we impose the visibility condition.
    \item Conditions for $I_P(i+t, t, {\it ccw})$ and $I_Q(j, t, {\it cw})$: $i+t=j$ and $p_{i+t}$ sees $p_i$ in $P$ and $A(i+t, t, {\it ccw}, {\it cw})$.
    \item Conditions for $I_P(i+t, t, {\it ccw})$ and $I_Q(k, t, {\it ccw})$: $i+t=k$ and $p_{i+t}$ sees $p_i$ in $P$ and $q_k$ sees $q_i$ in $Q$ and $A(i+t, t, {\it ccw}, {\it ccw})$.
\end{enumerate}

Since there are a quadratic number of subproblems, each taking constant time to solve, this algorithm runs in time $O(n^2)$, which proves:

\remove{ 
\attention{Anna: Sorry, but I could not quite understand the algorithm below, so I revised as above. In case there are advantages to the approach below, and someone can explain it better, please do.}

We are now ready to describe the dynamic programming algorithm. 
\attention{Anna: the notation $[i,j,cw]$ is above, but commented out.}
Let $\Chek[i,j,ccw]$ denote 
the result of the procedure to check whether there exist compatible paths such that the  path in $P$ 
starts at \suggestion{\sout{$p_j$}{$p_i$}}   and visits all the vertices with labels $[i,j,ccw]$. 
 We define $\Chek[i,j,cw]$ symmetrically. 


Using Claim~\ref{claim:interval}, it is straightforward to verify that  $\Chek[i,j,cw] = \Chek[i,j+1,cw]  \vee \Chek[j,i-1,ccw]$ if the following conditions hold.
\begin{enumerate}  
\item[-] $p_j$ is visible to both $p_{j+1}$ and $p_{i-1}$. Similarly, $q_j$ is visible to both $q_{j+1}$ and $q_{i-1}$. 
\question{Why do we need both visibility conditions? E.g., if I am extending the path by the edge $jj+1$, do I really need $p_j$ to be visible from $p_{i-1}$?}
\item[-] $q_{j+1}$ (similarly, $q_{i-1}$) is adjacent to either $q_i$ or $q_j$ on the boundary of $Q$. \question{Here the condition for $q_{j+1}$ is for $\Chek[i,j+1,cw]$ to be true, and the condition for $q_{i-1}$ is for $\Chek[j,i-1,ccw]$, correct?}
\end{enumerate}

We omit $\Chek[i,j+1,cw]$ (resp.,  $\Chek[j,i-1,ccw]$) if the corresponding visibility or adjacency conditions are not met. For example, $\Chek[j,i-1,ccw]$ is omitted in Figures~\ref{fig:dp-3}(a)--(b). If both checks get omitted, then $\Chek[i,j,ccw]$ returns `false'. We can define the recurrence relation for $\Chek[i,j,ccw]$ symmetrically. 
 
One can create a table of size $n\times n\times 2$ to store the solutions to the subproblems and implement the dynamic programming algorithm using table look-up. Note that the visibility and adjacency queries takes $O(1)$ time (with $O(n^2)$-time preprocessing). Since there are only $O(n^2)$ entries in the table, the running time is $O(n^2)$.
} 

\begin{theorem}
Given two $n$-vertex polygons, 
each with points labelled from $1$ to $n$ in some order,
one can find a pair of compatible paths or determine that none exist
in $O(n^2)$ time.
\end{theorem}

\subsubsection{A linear-time algorithm for convex point sets}

In this section we assume that the input is a pair of convex point sets $P, Q$, along with their convex hulls. 
 

Given a label $x$, we first define a \emph{greedy construction} to compute compatible paths starting at $x$. The output of the construction is an ordered sequence $\sigma_x$ of labels.  
Using Claim~\ref{claim:interval} we keep track of the intervals in $P$ and $Q$ corresponding to $\sigma_x$.
Initially $\sigma_x$ contains the label $x$.
Each subsequent step attempts to  add a new label to $\sigma_x$, maintaining intervals in $P$ and $Q$. 
Suppose the intervals corresponding to the current $\sigma_x$ are $I_P$ and $I_Q$ in $P$ and  $Q$ respectively.  Let $a$ and $b$ be the labels of the vertices just before and just after interval $I_P$ on the boundary of $P$.  Similarly, let $c$ and $d$ be the labels of the vertices just before and just after interval $I_Q$ on the boundary of $Q$.  If $\{a,b\} = \{c,d\}$, then we add $a$ and $b$ to $\sigma_x$ in arbitrary order.  Otherwise, if there is one label in common between the two sets, we add that label to $\sigma_x$.  Finally, if there are no common labels, then the construction ends.
%
%
Let $\sigma_x$ be a maximal sequence constructed as above.

\begin{lemma}~\label{lem:greedy}
$P$ and $Q$ have compatible paths starting at label $x$ if and only if $\sigma_x$ includes all $n$ labels.
\end{lemma}
\begin{proof}
If $P$ and $Q$ have compatible paths with label sequence $\sigma$ starting at label $x$ then by Claim~\ref{claim:interval} every prefix of $\sigma$ corresponds to an interval in $P$ and in $Q$, and we can build $\sigma_x$ in exactly the same order as $\sigma$.

For the other direction, we claim to construct noncrossing paths in $P$ and $Q$ corresponding to $\sigma_x$.  Observe that when we add one or two labels to $\sigma_x$,  we can add the corresponding vertices to our paths because the point sets are convex, so every edge is allowable.  
Furthermore, the paths constructed in this way are noncrossing because the greedy construction of $\sigma_x$ always maintains intervals in $P$ and $Q$. Hence the new edges are outside the convex hull of the paths so far.
\end{proof}

Lemma~\ref{lem:greedy} allows us to find compatible paths (if they exist) in $O(n^2)$ time by trying each label $x$ as the initial label of the path.
In order to improve this to linear time, we first argue that 
when $\sigma_x$ does not provide compatible paths, then we need not try any of its other labels as the initial label.

\begin{lemma}~\label{lem:compatible}
If $\sigma_x$ has length less than $n$, then no label in $\sigma_x$ can be the starting 
label for compatible paths of $P$ and $Q$.
\end{lemma}
\begin{proof}
Suppose that there are compatible paths with label sequence $s_y$ starting at a label $y$ in $\sigma_x$. 
Let $z$ be the first label that appears in $s_y$ but not in $\sigma_x$. Let $I_P$ and $I_Q$ be the intervals corresponding to $\sigma_x$ in $P$ and $Q$ respectively.
By Claim~\ref{claim:interval} the prefix of $s_y$ before $z$ corresponds to intervals, say
$I'_P$ and $I'_Q$ on $P$ and $Q$, respectively.  
Then $I'_P \subseteq I_P$ and $I'_Q \subseteq I_Q$ 
(by our assumption that $z$ is the first label of $s_y$ not in $\sigma_x$). 
Since the vertex with label $z$ must be adjacent to $I'_P$ on the boundary of $P$ and to $I'_Q$ on the boundary of $Q$, and $z$ does not appear in $\sigma_x$, therefore 
the vertex with label $z$ must be adjacent to $I_P$ on the boundary of $P$ and to $I_Q$ on the boundary of $Q$.  But then our construction would add label $z$ to $\sigma_x$.
\end{proof}

\remove{
\begin{lemma}~\label{lem:greedy}  
The longest compatible path starting at any vertex in $\sigma_x$ has length $|\sigma_x|$.
\end{lemma}
\begin{proof}
First observe that $\sigma_x$ provides compatible paths.  Here we use the fact that $P$ and $Q$ are convex. 
\attention{Maybe say more.}

We prove the rest of the Claim by contradiction.
Suppose that there is a longer compatible path $s_y$ starting at a label $y$ in $\sigma_x$. Let $z$ be the first label that appears in $s_y$ but not in $\sigma_x$. Let $I_P$ and $I_Q$ be the intervals corresponding to $\sigma_x$ in $P$ and $Q$ respectively. 
Let $I'_P$ and $I'_Q$ be the intervals corresponding to the prefix of $s_y$ before $z$.  Then $I'_P \subseteq I_P$ and $I'_Q \subseteq I_Q$.  Since the vertex with label $z$ must be adjacent to $I'_P$ on the boundary of $P$ and to $I'_Q$ on the boundary of $Q$, and $z$ does not appear in $\sigma_x$, therefore 
the vertex with label $z$ must be adjacent to $I_P$ on the boundary of $P$ and to $I_Q$ on the boundary of $Q$.  But then our construction would add label $z$ to $\sigma_x$.
\end{proof}
}

\remove{
\begin{lemma}~\label{lem:greedy}
The sequence $\sigma_x$ determines the 
\suggestion{longest} 
compatible paths (not necessarily spanning) among all the compatible paths that start with  label $x$.
\end{lemma}
\begin{proof}
We employ an induction on $n$, i.e., the number of points in each point set. The case when $n\le 3$ is straightforward. We thus assume that $n>3$, and the claim holds for all $n'$ where $n'<n$. Consider now the case when each point set has $n$ vertices. If $\sigma_x$ contains only one label, then the claim is straightforward to verify from  the construction. Otherwise, let $x,y$ be the first two labels in $\sigma_x$. By Claim~\ref{claim:interval}, $x,y$ must be adjacent on the convex hull of each point set. Let $v,w$ be the other labels adjacent to $x$ on the convex hull of $P$ and $Q$, respectively. 

If $v\not=w$, then we can apply an induction on $P\setminus{p_x}$ and $Q\setminus{q_x}$ to obtain the largest sequence $\sigma_y$. Since $y$ is the only choice from $x$ to construct compatible paths, $x,\sigma'_y$ would correspond to a pair of largest possible compatible paths starting at $x$. 

If $v=w$, then  we can apply an induction on $P\setminus{p_x}$ and $Q\setminus{q_x}$ to obtain the largest sequence $\sigma_y$. However, now it would create a problem if $x,\sigma_v$ is larger than $x,\sigma_y$. It now suffices to show that there always exists a sequence $x,\sigma_y$ with the same number of labels as that of $x,\sigma_v$. To prove this we modify  $x,\sigma_v$ as follows: If the $y$ is an endpoint of $x,\sigma_v$, then we can simply delete the edge $(x,v)$ and add the edge $(x,y)$. Otherwise, let $x,v,\ldots,v',y,y'\ldots,z$ be the sequence $x,\sigma_v$. In this case $v,\ldots,v'$ must appear consecutively in this order on the convex hull of each point set. Since the  point sets are convex, the sequence $x,y,v,\ldots,v',y'\ldots,z$ would be crossing free, and the required sequence $x,\sigma_y$. {\color{red} Need Some Polishing.)}
\end{proof}


\begin{lemma}~\label{lem:compatible}
If $\sigma_k$ does not contain all the labels, then no label in $\sigma_k$ can be the starting label for the compatible paths of $P$ and $Q$. 
\end{lemma}
\begin{proof}
Without loss of generality assume that $\sigma_k$ covers the interval $[i,j,cw]$ in $P$. Suppose for a contradiction that there exists a label $r\in [i,j,cw]$ such that $r\not=k$ and there exit compatible paths with starting label $r$. In other words, $\sigma_r$ contains all the labels.  

Let $L$ be the directed path determined by $\sigma_r$ in $P$. Let $(p_v,p_w)$ be the first edge on $L$, directed from $p_v$ to $p_w$, such that $w\not\in [i,j,cw]$. 


Assume first that $p_r$ coincides with $p_v$. If $r\not \in \{i,j\}$, then neither $p_i$ nor $p_j$ is visited, and thus the directed path $p_r (=p_v),p_w$ cannot be extended to a crossing-free path spanning all the vertices. If $r\in \{i,j\}$, then the vertices labelled $r$ and $w$ must appear consecutively on the convex hull of both $P$ and $Q$ (Claim~\ref{claim:interval}). Therefore, by the construction of $\sigma_k$, $w$ must belong to $\sigma_k$, a contradiction.  

Assume now that $r\not =v$. Without loss of generality assume that $r$ lies to the left of $p_vp_w$. In this case, the path $p_r,\ldots, p_v,p_w$ must cover all the vertices in $[v,w,cw]$ (Claim~\ref{claim:interval}), and thus also $p_j$. Note that $p_i$ does not belong to $p_r,\ldots, p_v,p_w$. In addition, since $w$ is the first label  (on $L$) outside $[i,j,cw]$, all the labels in $[v,w,cw] \setminus \{w\}$, must belong to $[v,j,cw]$. Therefore, the vertices labelled $j$ and $w$ must appear consecutively on the convex hull of both $P$ and $Q$. Therefore, by the construction of $\sigma_k$, $w$ must belong $\sigma_k$, a contradiction.  
\end{proof}
} 

We will use 
Lemma~\ref{lem:compatible} to show that we can eliminate some labels entirely when $\sigma_x$ is found to have length less than $n$. 
Suppose $\sigma_x$ does not include all labels. 
Let $I_P$ and $I_Q$ be the intervals on $P$ and $Q$, respectively, corresponding to the set of labels of $\sigma_x$. 
Let $a$ and $b$ be the labels that appear at the endpoints of $I_P$.

Suppose $P$ and $Q$ have compatible paths (of length $n$) with label sequence $\sigma$.  
Then by Lemma~\ref{lem:greedy} the initial and final label of $\sigma$ lie outside of $\sigma_x$.
Furthermore, by Claim~\ref{claim:consecutive}, the set of labels of $\sigma_x$ must appear consecutively and in the same order around $P$ and around $Q$ (either clockwise or counterclockwise). 
Our algorithm checks whether $I_P$ and $I_Q$ have the same ordered lists of labels.  If not, then there are no compatible paths.

So suppose that $I_P$ and $I_Q$ have the same ordered lists of labels.  Then the endpoints of $I_Q$ must have labels $a$ and $b$.
We will now reduce to a smaller problem by discarding all internal vertices of $I_P$ and $I_Q$.
Let $P'$ and $Q'$ be the point sets formed from $P$ and $Q$, respectively, by deleting the vertices with labels in $\sigma_x - \{a,b\}$.   

\begin{lemma}
\label{lem:reduce}  Suppose $z$ is a label appearing in $P'$.  
$P$ and $Q$ have compatible paths with initial label $z$  if and only if $P'$ and $Q'$ have compatible paths with initial label $z$. 
\end{lemma}
\begin{proof}
If $P$ and $Q$ have compatible paths (of length $n$) with initial label $z$, then we claim that deleting from those paths the vertices with labels in $\sigma_x - \{a,b\}$ yields compatible paths of $P'$ and $Q'$ with initial label $z$.  
It suffices to show that if we delete one vertex from a noncrossing path on points in convex position then the resulting path is still noncrossing.  The two edges incident to the point to be deleted form a triangle, and the new path will use the third side of the triangle.  Since the points are in convex position, the triangle is empty of other points, and so the new edge does not cross any other edge of the path.

For the other direction, suppose that $\sigma'$ is a label sequence of compatible paths of $P'$ and $Q'$ with initial label $z$.  Suppose without loss of generality that label $a$ comes before label $b$ in $\sigma'$.  Construct a sequence $\sigma$ by adding the labels of $\sigma_x - \{a,b\}$ after $a$ in $\sigma'$ in the order that they appear in $I_P$, e.g., see Figure~\ref{bottomup}. 
It remains to show that the corresponding paths in $P$ and $Q$ are noncrossing.  
This follows from the fact that in both $P$ and $Q$ the added points appear consecutively around the convex hull following the point with label $a$.
\end{proof}
\begin{figure*}[h]
\centering
\includegraphics[width=\textwidth]{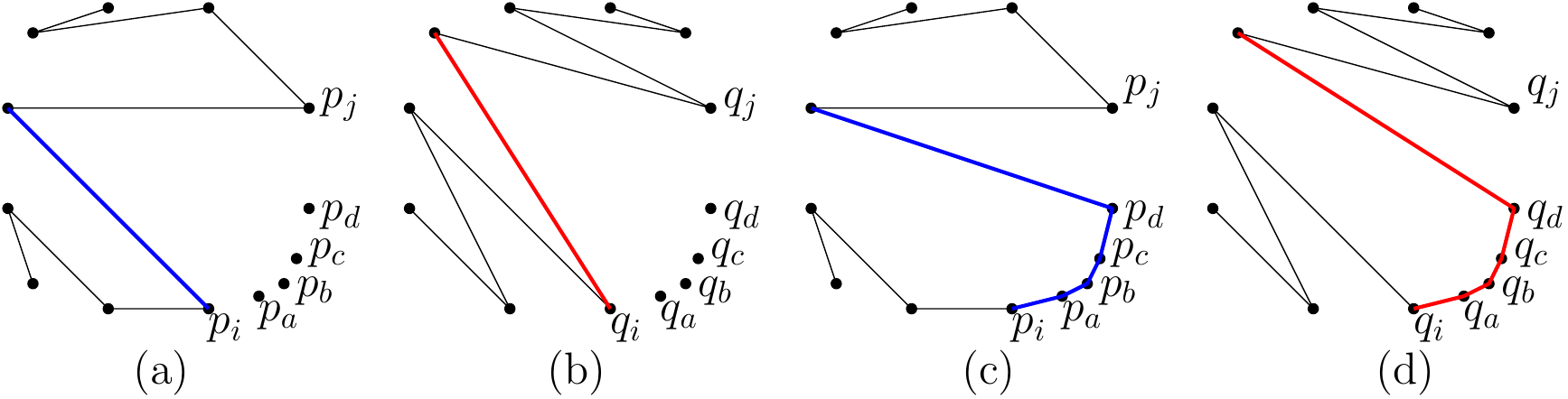}
\caption{(a)--(b)  Compatible paths on the point sets $P\setminus \{p_a,p_b,p_c,p_d\}$ and $Q\setminus \{q_a,q_b,q_c,q_d\}$. (c)--(d) Insertion of the deleted points keeps the paths compatible. 
}
\label{bottomup}
\end{figure*}
\remove{ 
Let $Q_k$ be all the points in $Q$ with labels in $\sigma_k$, and let $\lambda(Q,\sigma_k,cw)$ be the clockwise order of the points  in $Q_k$ on the convex-hull of $Q$. Similarly, define $\lambda(Q,\sigma_k,ccw)$. We now have the following lemma.

\begin{lemma}
\label{lem:ordering}
\attention{Anna: This lemma will vanish with above changes.}
Assume that $\sigma_k$ does not contain all the labels, but there exist compatible paths spanning $P$ and $Q$. Then the path $L_Q$ that spans $Q$ must contain the points of $Q_k$ in the same order as in $\lambda(Q,\sigma_k,cw)$ or in $\lambda(Q,\sigma_k,ccw)$. Note that the points of $Q_k$ may appear consecutively on the $L_Q$, i.e., they may appear as a subsequence.  The same condition holds also for the path $L_P$.
 
The paths $L_P$ and $L_Q$ can be modified such that they remain compatible, but all the points in $Q_k$ and $P_k$, except for one extreme point in each set, appear consecutively on the corresponding paths. {\color{red} Need figure - may be replace Fig~\ref{bottomup}.}
\end{lemma}
\begin{proof}
By Lemma~\ref{lem:compatible}, none of the labels in $\sigma_k$ can be a starting label for compatible paths. Assume that $\sigma_r$, where $r\not=k$, contains all the labels, i.e., determines compatible paths. Let $L_Q$ be the path spanning $Q$. 

Note that the lemma is straightforward to verify when $\sigma_k$ contains at most two labels. We may thus assume for a contradiction that $L_Q$ contains three points $\{q_a,q_b,q_c\}\subseteq Q_k$ in this order, which is different than the order of both $\lambda(Q,\sigma_k,cw)$ and $\lambda(Q,\sigma_k,ccw)$. The candidate $\lambda$-orders are $(q_b,q_a,q_c), (q_b,q_c,q_a)$ and their reverse orders. Since $L$ is crossing free, we only need to consider $(q_b,q_c,q_a)$ and its reverse order. In both of these cases, we must have the endpoint of $L_Q$ in $Q_k$, which contradicts that there do not exist compatible paths that start with a label from $\sigma_k$. 

We now modify $L_P$ and $L_Q$ such that the points in $Q_k$ and $P_k$ appear consecutively on the corresponding paths, as follows. Without loss of generality assume that $L_Q$ respects the order $\lambda(Q,\sigma_k,cw) = (q_a,\ldots,q_b,\ldots,q_c)$, where $q_a$ is the first point of $Q_k$ and $q_b,q_c$ are the last two points of $Q_k$. Then after visiting $q_a$, we visit all the points of   $Q_k$ upto $q_b$ and then continue visiting the points we skipped before $q_b$ in the order they appeared in $L_Q$ (Figure~\ref{?} {\color{red} Need figure}). Since the order of labels on $L_P$ is the same as that of $L_Q$, the same modification also works for $L_P$.
\end{proof}
} 

We can now prove the main result of this section.

\begin{theorem}
Given two sets of $n$ points in convex position (along with their convex hulls)
each with points labelled from $1$ to $n$,
one can find a pair of compatible paths or determine that none exist
in linear time.
\end{theorem}
\begin{proof}
The algorithm is as described above.  At each stage we try some label $x$ to be the initial label of compatible paths, by computing 
$\sigma_x$ using the greedy construction.  If $\sigma_x$ has length $n$ we are done.  Otherwise if $\sigma_x$ has length 1 or 2, then we have ruled out the labels in $\sigma_x$ as initial labels.  Finally, if $\sigma_x$ has length less than $n$ and at least 3 then we 
test whether the intervals corresponding to $\sigma_x$ in $P$ and $Q$ have the same ordering, and if they do, then we
apply the reduction described above and recurse on the smaller instance as justified by Lemma~\ref{lem:reduce}. 

The running time of the algorithm is determined by the length of all the $\sigma$-sequences we compute. Define a $\sigma$-sequence to be `long' or `short' depending on whether it contains at least three labels or not. Every long sequence of length $\ell$ reduces the number of points by $(\ell-2)$ and requires $O(\ell)$ time.  Thus, long sequences take $O(n)$ time in total. Computing any short sequence takes $O(1)$ time. Since for each label, we compute $\sigma$ at most once, the short sequences also take  $O(n)$  time in total.
\end{proof}

\remove{ 
By Lemma~\ref{lem:greedy}, if the point sets admit compatible paths with the starting label $x$, then $\sigma_x$ would contain all the labels and determine such compatible paths. Therefore, it suffices to repeatedly check $\sigma_k$ for every label, which may take $O(n^2)$ time.

To reduce the running time we use Lemmas~\ref{lem:compatible}--\ref{lem:ordering}.  By Lemma~\ref{lem:compatible}, if the point sets do not admit compatible paths with the starting label $k$, then none of the labels in $\sigma_k$ can be starting labels. Therefore, every time we compute some $\sigma_k$, for some label $k$, we can check whether  $\lambda(Q,\sigma_k,cw)$ or its reverse coincides with  either $\lambda(P,\sigma_k,ccw)$ or its reverse. If not, then by Lemma~\ref{lem:ordering}, the point sets do not admit compatible paths. Otherwise, we delete all the points except for the extreme points of $Q_k$ and $P_k$,  and repeat the process of computing $\sigma$ for a new label.  

The reason we keep the extreme points of $Q_k$ and $P_k$ is because they control the `compatibility' ordering of the paths. Any compatible paths that we compute can be extended to add the deleted points in the order determined by those extreme points (Figure~\ref{bottomup} {\color{red} Add convexity arguments in the Appendix}). Observe that the construction will only yield  compatible paths that contain the points $P_x$ and $Q_x$ consecutively on the convex hull of $P$ and $Q$, respectively. However, this would suffice since by Lemma~\ref{lem:ordering}, every pair of compatible paths can be modified to satisfy this property.  

The running time of the algorithm is determined by the length of all $\sigma$-sequences we compute. Define a $\sigma$-sequence to be `long' or `short' depending on whether it contains at least three labels or not. Every long sequence of length $\ell$ reduces the number of points by $(\ell-2)$ and requires $O(\ell)$ time.  Therefore, large sequences take $O(n)$ time in total. Computing every short sequences takes $O(1)$ time. Since for each label, we compute $\sigma$ at most once, the short sequences also take  $O(n)$  time in total.
} 

\section{Monotone Paths in General Point Sets}

In this section we examine arbitrary point sets in general position, but we restrict the type of path. 

Let $P$ be a point set in general position. An ordering $\sigma$ of the points of $P$ is called \emph{monotone} if there exists some line $\ell$ such that the orthogonal projection of the points on $\ell$ yields the order $\sigma$. A \emph{monotone path} is a 
path that corresponds to a  monotone ordering.  
Note that every monotone path is noncrossing.

Two points sets $P$ and $Q$ each labelled $1, 2, \ldots, n$ have \emph{compatible monotone paths} if there is an ordering of the labels that corresponds to a monotone path in $P$ and a monotone path in $Q$.   
To decide whether compatible monotone paths exist, we can enumerate all the monotone orderings of $P$, and for each of them check in linear time whether it determines a monotone path in $Q$.

A method for enumerating all the monotone orderings of a point set $P$ was developed by Goodman and Pollack:  

\begin{theorem} [Goodman and Pollack~\cite{goodman1980combinatorial}]
Let $\ell_0$ be a line 
not orthogonal to any line determined by two points of $P$. 
Starting with $\ell=\ell_0$, 
rotate the line $\ell$ through $360^\circ$ counter-clockwise about a fixed point.
Projecting the points onto $\ell$ as it rotates 
gives all the possible monotone orderings of $P$.
There are $2 {n \choose 2}= n(n-1)$ orderings, and each successive ordering differs from the previous one 
by a swap of two elements adjacent in the ordering.
\label{theorem:ordering}
\end{theorem}

Furthermore, the sequence of swaps that change each ordering to the next one can be found in $O(n^2\log n)$ time by sorting the $O(n^2)$ lines (determined by all pairs of points) by their slopes. 

This gives a straight-forward $O(n^3)$ time algorithm to find compatible monotone paths, since we can generate the $O(n^2)$ monotone orderings of $P$ in constant time per ordering, and check each one for monotonicity in $Q$ in linear time.

We now present a more efficient $O(n^2 \log n)$ time algorithm.
For ease of notation, relabel the points so that the order of points $P$ along $\ell_0$ is $1, 2,\dots,n$.
As the line $\ell$ rotates, let $L^P_0, L^P_1, \ldots L^P_{t-1}$, where $t=n(n-1)$, be the monotone orderings of $P$, 
and let $S_P$ be the corresponding swap sequence.
Similarly, let $L^Q_0, L^Q_1, \ldots L^Q_{t-1}$ be the monotone orderings of $Q$ and let $S_Q$ be the corresponding swap sequence (Figure~\ref{fig:rotation}).  
We need to find whether there exist some $i$ and $j$ such that $L^P_i = L^Q_j$. 
As noted above, $S_P$ and $S_Q$ have size $O(n^2)$ and can be computed in time $O(n^2 \log n)$.

\remove{ older version -- what has changed is the order of presentation: first Goodman-Pollack, then the alg. via duality and arrangements, then 
the fact that we get an n^3 algorithm, then segue to our faster alg.

This gives a straight-forward $O(n^3)$ time algorithm.
 In this section we give an $O(n^2)$ time algorithm.


We first describe some preliminary results, which will help describing our algorithm.

\begin{theorem} [Goodman and Pollack~\cite{goodman1980combinatorial}]
Given a configuration $\it{C}$  of $n$ points in general position labelled 
$1, 2, \dots,n$ and a directed line $\ell$ which is not orthogonal to any line determined by two points of $\it{C}$, the orthogonal projection of $\it{C}$ on $\ell$ determines a permutation of $1, 2, \dots,n$ in an obvious way. 
\attention{Note that we get all monotone orderings.}
As the line $\ell$ rotates counter-clockwise about a fixed point we obtain a sequence of permutations of period $2 {n \choose 2}= n(n-1)$. Moreover, successive permutations differ only by having the order of two adjacent labels switched. 
\label{theorem:ordering}
\end{theorem}

Theorem \ref{theorem:ordering} gives us a natural way to solve the  compatible monotone paths problem. Choose some starting line $\ell_0$ in $P$, and relabel the points in $P$ and $Q$ such that the order of the points of $P$ along $\ell_0$ is $1, 2,\dots,n$. As the line $\ell_0$ rotates about a fixed point, 
 we obtain $n(n-1)$ different orderings of $P$, e.g., see Figure~\ref{fig:rotation}.  Call these orderings $L^P_0, L^P_1, \ldots L^P_{t-1}$, where $t=n(n-1)$.  Each $L^P_i$ differs from the previous one by a single swap of two adjacent elements, 
 \suggestion{and $L^P_{n \choose 2}$ is the reverse list $n, n-1, \ldots, 1$.} 
 Note that we do not explicitly compute all these orderings (which would take $O(n^3)$ time), but compute a list $S_P$ of swaps that implicitly corresponds to $L^P_0, L^P_1, \ldots L^P_{t-1}$. It is straightforward to construct $S_P$ in $O(n^2)$ time by using the line arrangement determined by the pairs of points in $P$~\cite{goodman1980combinatorial}. 
 Similarly, for $Q$, we can construct a list $S_Q$ for the orderings  $L^Q_0, L^Q_1, \ldots L^Q_{t-1}$.  We now need to see whether there exist some $i$ and $j$ such that $L^P_i = L^Q_j$. This is nontrivial since $S_P$ and $S_Q$ are implicit, and we want to keep the time complexity within $O(n^2)$.
 }

\begin{figure} 
\centering
\includegraphics[width=.7\textwidth]{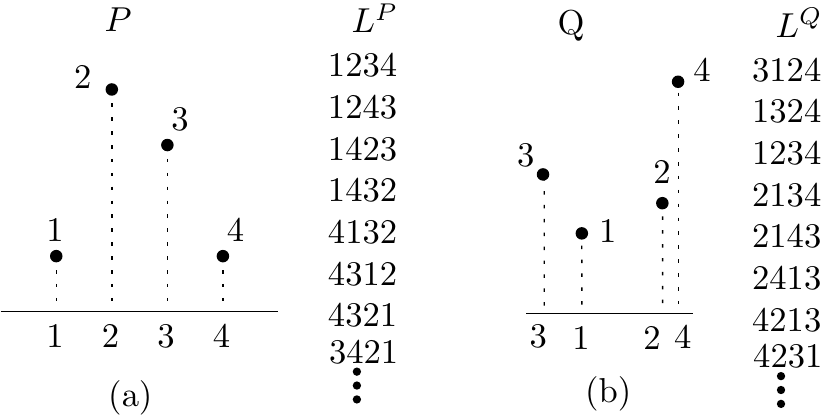}
\caption{Illustration for computing compatible monotone paths.
}
\label{fig:rotation}
\end{figure}
 
Recall that the \emph{inversion number}, $I(L)$ of a permutation $L$ is the number of pairs that are out of order. 
 It is easy to see that the inversion numbers of the $L^P_i$'s progress from $0$ to ${n \choose 2}$  and back again.
 In particular, $I(L^P_i) = i$ for $0 \le i \le {n \choose 2}$.
 Our algorithm will compute the inversion numbers of the $L^Q_j$'s, which also have some structure. 
 Let $I_j$ be the inversion number of $L^Q_j$.  Note that we can compute $I_0$ in $O(n \log n)$ time---sorting algorithms can be modified to do this~\cite{DBLP:books/daglib/0023376}.
 
 \begin{claim}
 For all $j$, $1 {\le} j {\le} n(n-1)$, $I_j$ differs from $I_{j-1}$ by $\pm 1$, and can be computed from $I_{j-1}$ in constant time.
 \label{claim:inversions}
 \end{claim}
 \begin{proof}
 $L^Q_j$ is formed by swapping one pair of adjacent elements in $L^Q_{j-1}$.  If this swap moves a smaller element after a larger one then $I_j {=} I_{j-1} {+} 1$.  Otherwise, it is $I_{j-1} {-} 1$.
 \end{proof}
 
 The main idea of our algorithm is as follows.  If $L^Q_j = L^P_i$,  then they must have the same inversion number, $I_j$.
 There is one value of $i$ in the range $0 \le i < {n \choose 2}$ that gives this inversion number, namely $i = I_j$.  
 There is also one value of $i$ in the second half of the range that gives this inversion number, but 
 we can ignore the second half of the range based on the following:
 
 \begin{remark}\label{r:half}
If there exist $i,j$ such that $L^P_i = L^Q_j$, then 
there is such a pair with
$i$ in the first half of the index range, i.e., $0 \le i < {n \choose 2}$.
\end{remark}
\begin{proof}
The second half of each list of orderings contains the reversals of the orderings in the first half~\cite{goodman1980combinatorial}.  
Thus if 
there is a match $L^P_i = L^Q_j$ then the reversals of the two orderings also provide a match, say $L^P_{i'} = L^Q_{j'}$, and either $i$ or $i'$ is in the first half of the index range. 
\end{proof}

Our plan is to iterate through the orderings $L^Q_j$ for $0 \le j < n(n-1)$.  Since each ordering differs from the previous one by a single swap,  we can update from one to the next in constant time. For each $j$, we will check if $L^Q_j$ is equal to $L^P_{I_j}$, i.e., 
for each $j$, $0 \le j < n(n-1)$ we will compute the following four things:
\begin{itemize}
    \item $L^Q_j$, $I_j$, $L^P_{I_j}$, and 
    \item $H_j$, which is 
    the \emph{Hamming distance}---i.e., the number of mismatches---between $L^Q_j$ and $L^P_{I_j}$ 
\end{itemize}

If we find a $j$ with $H_j=0$ then we output $L^Q_j$ and $L^P_{I_j}$ as compatible monotone paths.  
Otherwise, we declare that no compatible monotone paths exist.  Correctness of this algorithm follows from Remark~\ref{r:half} and the discussion above:

\begin{claim}
$P$ and $Q$ have compatible monotone paths if and only if $H_j=0$ for some $j$, $0 \le j < n(n-1)$.
\end{claim}

We now give the details of how to perform the above computations.  For $j=0$ we will compute everything directly, and for each successive $j$, we will show how to update efficiently.
We initialize the algorithm at $j=0$ by computing $L^Q_0$ and $I_j$ in $O(n \log n)$ time, $L^P_{I_j}$ in $O(n^2)$ time, and $H_j$ in linear time.

Now consider an update from $j-1$ to $j$.  
As already mentioned, $L^Q_j$ differs from $L^Q_{j-1}$ by one swap of adjacent elements, so we can update in constant time.
By Lemma~\ref{claim:inversions}, $I_j$ differs from $I_{j-1}$ by $\pm 1$ and we can compute it in constant time.
This also means that $L^P_{I_j}$ differs from $L^P_{I_{j-1}}$ by one swap of adjacent elements, so we can update it in constant time.

Finally, we can update the Hamming distance in a two-step process as the two orderings change.  When we update from $L^Q_{j-1}$ to $L^Q_j$, two positions in the list change, and we can compare them to the same positions in $L^P_{I_{j-1}}$ to update from $H_{j-1}$ to obtain the number of mismatches between $L^Q_j$ and $L^P_{I_{j-1}}$.  
 When we update to $L^P_{I_j}$, two positions in this list change, and we can compare them to the same positions in $L^Q_{I_j}$ to update to $H_j$.  This two-step process takes constant time.

In total, we spend $O(n^2)$ time on initialization and constant time on each of $O(n^2)$ updates, for a total of $O(n^2)$ time.
%
\remove{
\attention{Here is the previous version:}

We first explore some properties of $S_P$ and $S_Q$.
 The following remark follows from the construction in~\cite{goodman1980combinatorial}.
 
\begin{remark}
Consider the cyclic order of $S_P$. Once a pair of labels are swapped, they are not swapped again until we compute all the $\binom n2$ distinct swaps.
\attention{Should rewrite this and/or combine with remark below.}
\label{prop:swap_inversion}
\end{remark}

The second half of the list $L^P_0, L^P_1, \ldots L^P_{t-1}$ contains the reversals of orderings in the first half~\cite{goodman1980combinatorial}, which implies the following remark.

\begin{remark}\label{r:half}
If there exist $i,j$ such that $L^P_i = L^Q_j$, then we can find such $i,j$ such that $i$ is in the first half of the index range, i.e., $0 \le i < {n \choose 2}$.
\end{remark}

\suggestion{Define the \emph{inversion distance}, $d_I$, between two orderings $L$ and $L'$ of $1, \ldots , n$ to be the number of pairs that appear in 
opposite order in $L$ and $L'$. 
For example $d_I((1,2,3,4), (2,3,1,4)) = 2$ since the pairs $1,2$ and $1,3$ are in opposite order.
Inversion distance counts the number of swaps of adjacent elements 
that are required to turn one ordering into the other.
}
The inversion distance between two $n$-element lists can be computed in $O(n \sqrt{\log n})$ time~\cite{Chan:2010:CIO:1873601.1873616}. 

\begin{lemma}
\label{lemma:inversion_distance}
For each $i$,  $0 \le i < {n \choose 2}$, the inversion distance between $L^P_i$ and $L^P_0$ is $i$. 
\end{lemma}
\begin{proof}
 By Remark~\ref{prop:swap_inversion}, we can transform $L^P_0$ into $L^P_i$ by a sequence of $i$ unique adjacent swaps. This is exactly the number of pairs that appear in opposite order in the two lists, and thus $d_I(L^P_0,L^P_i) = i$.
\end{proof}

By Remark~\ref{r:half}, it suffices to restrict our attention to $L^P_i$, where $0 \le i < {n \choose 2}$. 
\suggestion{Using Lemma~\ref{lemma:inversion_distance} we obtain:}
\begin{cor}
If $L^Q_j$ and $L^P_i$ correspond to the same permutation, for some $0 \le i < {n \choose 2}$ and $0\le j < n(n-1)$, then $d_I(L^Q_j, L^P_1) = d_I(L^P_i, L^P_1) = i$.  
\label{cor:match_inversion}
\end{cor}

\subsection{Algorithm}

\suggestion{The idea is to iterate over lists $L^Q_j$ for $0 \le j < n(n-1)$, and for each $j$, to compute the following two quantities:} 
\begin{itemize}
    \item $d_I(L^Q_j, L^P_0)$, which we will refer to as $D(j)$.
    \item \suggestion{$d_H(L^Q_j, L^P_{D(j)})$, where $d_H$ is the Hamming distance, so this counts the number of mismatches between $L^Q_j$ and $L^P_{D(j)}$.}
\end{itemize}

We will show that 
monotone compatible paths exist 
if and only if there exists some $j$ such that  $d_H(L^Q_j, L^P_{D(j)})$ equals  zero. 
Note that we first compute $D(0)$,  which takes $O(n \sqrt{\log n})$ time~\cite{Chan:2010:CIO:1873601.1873616}, but for subsequent values of $j$, we will show how to compute the inversion distance in constant time (Lemma~\ref{lem:dj}). 
Similarly, computing $d_H$ takes $O(n)$ time 
\suggestion{for $j=0$, but for subsequent values of $j$, we will show how to computer $d_H$ in constant time (Lemma~\ref{lem:ham}).}
Thus the algorithm will have a run time of $O(n^2)$.

We first show that for $j\ge 1$, $d_I(L^Q_j, L^P_0)$ 
\suggestion{differs from $d_I(L^Q_{j-1}, L^P_0)$ by $\pm 1$.} 
This property helps us to compute $D(j)$ in $O(1)$ time.

\begin{lemma}\label{lem:dj}
 For $j\ge 1$, $D(j)$ must be either $D(j-1) + 1$ or $D(j-1) - 1$. Furthermore, we can compute $D(j)$ from $D(j-1)$ in $O(1)$ time.
\end{lemma}
\begin{proof}
First note that if $D(j)$ is equal to $D(j-1)$, then $d_I(L^Q_{j}, L^P_1) = d_I(L^Q_{j-1}, L^P_1)$. However, by Remark~\ref{prop:swap_inversion}, two consecutive permutations in $L^Q$ cannot have the same inversion distance from any other permutation in $L^Q$ (assuming $n\ge 3$). Thus $D(j)$ cannot be  equal to $D(j-1)$. 

On the other hand, the absolute difference between $D(j)$ and $D(j-1)$ cannot be larger than 1, since the inversion distance between $L^Q_j$ and $L^Q_{j-1}$ is 1 (i.e., they differ by only one swap). 


To compute $D(j)$, we use the fact that $L^Q_j$ and $L^Q_{j-1}$ differ by exactly one swap of consecutive elements. Recall that we relabelled the points in $P$ such that $L^P_0$ equals $1, 2, \dots, n$. Therefore, $D(j)$ will be larger or smaller than $D(j-1)$ based on whether we are swapping (from left to right) a smaller value with a larger value, or a larger value with a smaller value.
\end{proof}

\begin{lemma}\label{lem:ham}
For $j\ge 1$, $d_H(L^Q_j, L^P_{D(j)})$ can be computed in constant time.
\end{lemma}
\begin{proof}
We maintain two sequences, $A^P$ and $A^Q$, to store the the current permutations $L^Q_j$ and $L^P_{D(j)}$,  respectively. We will also maintain a counter to keep track of the number of mismatches. 

\attention{Anna: I'm redoing indices and got to here.  To be continued.}

We compute $d_H(L^Q_1, L^P_{D(1)+1})$ explicitly in $O(n)$ time by comparing the corresponding elements  in both the sequences. For each $j$ from $2$ to $n(n-1)$, we compute $d_H(L^Q_j, L^P_{D(j)+1})$ in two phases, as follows. 

In the first phase we compute $d_H(L^Q_j, L^P_{D(j-1)+1})$ by performing the swap on $L^Q_{j-1}$, which was stored in $A^Q$, to obtain $L^Q_{j}$. Thus $A^Q$ now contains $L^Q_j$. At the same time, we compute the number of mismatches between $L^Q_j$  and  $L^P_{D(j-1)+1}$. Assume that 
 we swapped the $k$-th and ${k+1}$-th elements of $L^Q_{j-1}$. Then it is straightforward to update  the number of mismatches by comparing the swapped elements with the $k$-th and $(k+1)$-th elements of $L^P_{D(j-1)+1}$.


In the second phase we compute $D(j)$ (using Lemma~\ref{lem:dj}), and then perform the necessary swap on $L^P_{D(j-1)+1}$, which was stored in $A^P$,  to obtain $L^P_{D(j)+1}$. Thus $A^P$ now contains $L^Q_j$. At the same time, we compute $d_H(L^Q_j, L^P_{\suggestion{D(j)}+1})$ following the idea we used in the first phase, i.e., examining the positions of the swapped elements in both sequences. 

Since we only need a constant number of positions to examine in $A^P$ and $A^Q$, the updates can be carried out in $O(1)$ time. 
\end{proof}

It is immediate from Lemmas~\ref{lem:dj}--\ref{lem:ham} that the algorithm takes $O(n^2)$ time. Thus the only concern is whether there exists some $i,j$, where the algorithm fails to notice $L^Q_j = L^P_i$. Since we iterate over the all the permutations in $L^Q$, but only the first half of the permutations in $L^P$, this may happen for some $L^Q_j$. However, by Remark~\ref{r:half}, there would be some $L^Q_{j'}$ that coincides with some permutation of the first half of $L^P$ (here $L^Q_{j'}$ is in fact the reverse   of $L^Q_j$).

}  
We thus obtain the following theorem. 
\begin{theorem}
Given two point sets, each containing $n$ points labelled from $1$ to $n$, 
one can find a pair of compatible monotone paths or determine that none exist
in $O(n^2 \log n)$ time.
\end{theorem}

\jyoti{\section{Conclusion}
In this paper we gave some fast polynomial-time algorithms to find compatible paths (if they exist) on a given pair of labelled point sets under some  constraints on the paths or the point sets. An interesting direction for future  research would be to extend our linear-time algorithm to point sets that determine $k>1$  nested convex hulls.  Since the problem  is NP-complete~\cite{DBLP:conf/isaac/HuiS04} in general, it would also be interesting to examine fixed-parameter tractable algorithms.}  

\subsection*{Acknowledgement} We thank the organizers of the Fields Workshop on Discrete and Computational Geometry, held in July 2017 at Carleton University. 
E.~Arseneva is supported in part by the SNF Early Postdoc Mobility grant P2TIP2-168563 and by F.R.S.-FNRS; 
A.~Biniaz, K.~Jain, A.~Lubiw, and D.~Mondal are supported in part by NSERC.

\bibliographystyle{abbrv}
\bibliography{main}

\end{document}